\documentclass[12pt]{article}

\usepackage{adjustbox}
\usepackage[linesnumbered,ruled,vlined]{algorithm2e}
\usepackage{amsfonts}
\usepackage{amsmath}
\usepackage{amssymb}
\usepackage{amsthm}
\usepackage{array}
\usepackage{color}
\usepackage{enumerate}
\usepackage{epsfig}
\usepackage[mathscr]{euscript}
\usepackage{graphicx}
\usepackage{hyperref}
\usepackage{subfigure}
\usepackage{ulem}
\usepackage{url} 
\usepackage{varioref}
\usepackage{xcolor}
\usepackage{hyperref}
\usepackage{wrapfig}

\SetKwInput{KwInput}{Input}                
\SetKwInput{KwOutput}{Output}              

\newtheorem{definition}{Definition}
\newtheorem{theorem}{Theorem}
\newtheorem{remark}{Remark}
\newtheorem{proposition}[theorem]{Proposition}

\begin{document}
	
	\begin{center}
		\large{\textbf{Clustering Sets of Functional Data by Similarity in Law}}
		\end{center}
		\small {\textbf{Dedicated to the memory of Antonio Galves}}
	\begin{center}
		Antonio Galves$^{1,\dagger}$, Fernando Najman$^2$, Marcela Svarc$^{3,4}$, and Claudia D. Vargas$^5$
	\end{center}		
	
	$^{1}$Instituto de Matem\'atica e Estat\'istica, Universidade de São Paulo, São Paulo, Brazil, $\dagger$ Deceased

	$^{2}$Instituto de Computação, Universidade Estadual de Campinas, Campinas, Brazil, fanajman@gmail.com
	
	$^{3}$Departamento de Matem\'atica y Ciencias, Universidad de San Andr\'es, Argentina ,
	msvarc@udesa.edu.ar
	
	$^{4}$CONICET
	
	$^{5}$Instituto de Biof\'isica Carlos Chagas Filho, Universidade Federal do Rio de Janeiro, Rio de Janeiro, Brazil.,
	cdvargas@biof.ufrj.br

	\begin{abstract}
		We introduce a new clustering method for the classification of functional data sets by their probabilistic law, that is, a procedure that aims to assign data sets to the same cluster if and only if the data were generated with the same underlying distribution. This method has the nice virtue of being non-supervised and non-parametric, allowing for exploratory investigation with few assumptions about the data. Rigorous finite bounds on the classification error are given along with an objective heuristic that consistently selects the best partition in a data-driven manner. Simulated data has been clustered with this procedure to show the performance of the method with different parametric model classes of functional data. 
	\end{abstract}
	\textit{Keywords:} Kolmogorov Smirnov statistics, Concentration inequalities, Projection procedure.
	
	\section{Introduction}
	
	An important problem in functional data analysis is to identify among different data sets those generated by the same probabilistic law. To solve this problem, we propose a novel method that clusters sets of functional data by their similarity in law. The method is non-parametric and non-supervised giving it a broad range of applications. The procedure measures the similarity of the data sets empirical distributions and uses this similarity measure to conduct a hierarchical clustering procedure. The partition of the data sets results in a label for each data set which represents estimated equivalence in law. While there are diverse options proposed for clustering functional data samples, clustering procedures for sets of functional data are a less discussed topic.  Clustering sets of real data by their distribution was discussed in Mora-L\'opez and Mora \cite{mora2015adaptive} and Zhu et al. \cite{zhu2021clustering}. However, methods for higher dimensional data cannot be trivially extended from the one-dimensional case.

	To measure the similarity in law between each pair of data sets we used a random projection strategy. We first project each individual functional data sample associated with each data set into a fixed number of randomly generated directions. Then, for each pair of data sets and each random direction, we measure a distance between the one-dimensional empirical distributions distance between the real-valued projections. Our method uses these distances to construct an estimation of the distance of the probabilistic law generating the functional samples. This approach is inspired by the results of Cuesta-Albertos, Fraiman and Ransford \cite{Cuesta2007}, where they show conditions under which the distribution of the projections of the samples of two data sets will be equal if and only if the two laws are the same.

	We also provide finite bounds for the probability of encountering a large error obtained by approximating the distance between laws by the estimate obtained with a finite sample. A non-parametric bound for multivariate data has been presented in Naaman \cite{naaman2021tight}, however, the bound increases linearly with the dimensionality of the data, making it undefined for the functional case. The bound presented here holds for data in Hilbert spaces, i.e. it does not depend on the data dimensionality. The bound also does not assume that the data be generated by any parametric model class. Without loss of generality, we assume in the following that the data are functions living in $L^2([0,T])$, for some fixed real $T > 0$. We also present an adapted version of the bound under $H_0$ designed to be sharper, and show the clustering procedure good performance with functional data simulated with models from two different parametric model classes.

	Informally, the algorithm works as follows. Let $\mathcal{U}$ be a finite set such that $ S_U  = |\mathcal{U}| \geq 2$, and let also $\mathcal{Y}_N^u$ be a family of sets with $N$ functional data samples indexed by $u \in \mathcal{U}$, that is, $\mathcal{Y}_N^u = (Y^u_1, \cdots, Y^u_N)$ with $Y^u_n \in L^2([0,T])$. For simplicity, we assume that all data sets have exactly $N$  samples, but we note that all steps in the procedure can naturally be adapted for sets with different sample sizes. For each $u \in \mathcal{U}$, we denote by $Q^u$ the probabilistic law that generates the samples of the data set $\mathcal{Y}_N^u$. Let also $(B_1, \cdots, B_M)$ be $M$ random directions in $L^2([0,T])$ generated in a suitable way. Using a random projection strategy, we construct an estimated distance $\hat{D}_{N,M}(u,v)$ between the empirical distributions of the $\mathcal{Y}^u_N$ and $\mathcal{Y}^v_N$ functional data sets. We use this estimated distance as a dissimilarity measure for a hierarchical clustering procedure to partition the functional data sets indexed by $u \in \mathcal{U}$ by their law, that is, our goal is to assign two data sets $\mathcal{Y}_N^u$ and $\mathcal{Y}_N^v$ to the same cluster if and only if $Q^u = Q^v$. 
	
	 \section{Clustering procedure} \label{sec:clustproc}

	We propose to use the following hierarchical clustering procedure to retrieve a partition of the functional data sets. Let $\mathcal{P}$ be the set of all partitions of the set $\mathcal{U}$. Given $P \in \mathcal{P}$, let $D$ be some dissimilarity between elements of $\mathcal{U}$. We also define a dissimilarity $D$ for any pair of clusters which by abuse of notation we also denote by $D$. This strategy of extending the dissimilarity to cluster pairs is usually called the \textit{linkage} of the hierarchical clustering procedure. Here we propose to use the \textit{complete linkage}, that is, for any pair of clusters $C$ and $C'$, we define the linkage as
	
	$$
	D(C,C') = \sup_{u \in C, v \in C'}\{D(u,v)\} \ .
	$$
	
	Denote $C_1(P)$ and $C_2(P)$ as two elements of $P$ satisfying $C_1(P) \neq C_2(P)$ and
	$$
	D(C_1(P), C_2(P)) \leq D(C, C'), \mbox{ for all pairs } \ C \in P, \ C' \in P, \ C \neq C' \ .
	$$
	
	In the following, we assume that all pairs have different dissimilarity values such that $C_1(P)$ and $C_2(P)$ are uniquely defined. This will be the case with probability one for our chosen dissimilarity, which will be introduced later and which takes real values.
	
	Let $r_{D}(P)$ be the dissimilarity between $C_1(P)$ and $C_2(P)$, i.e.,
	\begin{equation}\label{eq:DCompLink}
		r_{D}(P) = \inf \left\{ D(C, C'): (C, C') \in P \times P, C \cap C' = \emptyset \right\} \ .
	\end{equation}
	Let us denote $C_{1,2}(P) = \{C_1(P) \bigcup C_2(P)\}$. Consider the following family of recursive partitions. Let
	$$
	P_1(D) = \{\{u\} : u \in \mathcal{U}\}
	$$
	be the partition of singletons, and for $k = 2, \ldots, S_U$ let
	$$
	P_k(D) = \left\{C \in P_{k-1} : C \neq C_1(P_{k-1}), C \neq C_2(P_{k-1}) \right\} \bigcup C_{1,2}(P_{k-1}) \ .
	$$

	We call a \textit{dendrogram model} the pair consisting of the family $P_{1:S_U}(D) = (P_1(D), \cdots, P_{S_U}(D))$ and the associated function $r_{D}$. Dendrogram models can be represented graphically as a rooted and labelled tree. An example of a graphical representation of a dendrogram model is shown in Fig. \ref{fig:ExDen}.
	
	\begin{figure}
		\centering
		\includegraphics[scale = 0.25]{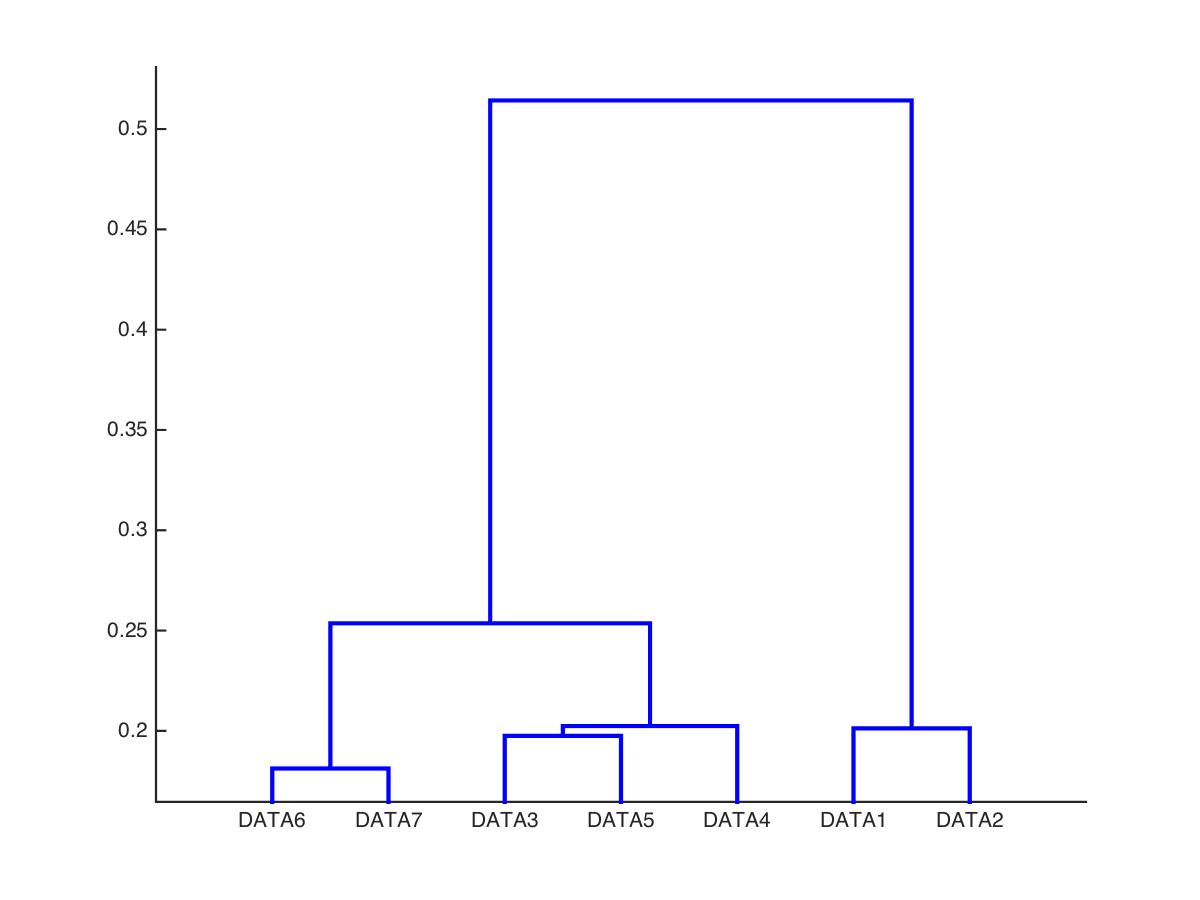}
		\caption{Graphical representation of a dendrogram model. Each labelled leaf represents a data set labelled by an element of $\mathcal{U}$. The height at which two clusters are connected represents the correspondent value of $r_D$.}
		\label{fig:ExDen}
	\end{figure}

	Our goal is to obtain a plug-in $D$ dissimilarity measure for data sets indexed by the elements of $\mathcal{U}$. In addition, we want to choose a measure that gives strong and consistent estimates of the similarity in law between the sets. In the following, we present a measure with such properties. 
	
	\section{Preliminary definitions}

	We propose a dissimilarity function for functional data sets based on random projections which estimates the similarity between the laws associated with these data sets. Let us start by giving some proper definitions to present Proposition \ref{prop:distance}. Let $\mathcal{U}$ be a finite set with cardinality denoted by $S_{\mathcal{U}}$. For some fixed real value $T > 0$ and probability space ($\Omega, \mathcal{A}, \mathbb{P}$), let $F = L^2([0, T])$ and let $\mathcal{F}$ be the Borel $\sigma$-algebra on $F$. For all $u \in \mathcal{U}$, let $Q^u$ be a probability measure on $L^2[0,T]$ and $Y^u:\Omega \rightarrow L^2[0,T]$ be  a random element of $L^2[0,T]$ generated with distribution $Q^u.$ Moreover, let $h$ be an element of the dual space of $F,$ which is $L^2[0,T].$ Denote $Q^u_h$ to the univariate distribution of the random variable $\left< Y^u,h \right>$ and $f^{u,h}$ to its cumulative distribution function.

	\begin{definition}[Carleman condition]\label{def:carleman}
		We say that a probability measure $Q$ on $(F, \mathcal{F})$ satisfies the \textit{Carleman condition} if, for all $i >0$, its absolute moments $m_i = \int||h||Q(dh)$ are finite and
		$$
		\sum_{j \geq 1}m_i^{-1/j} = +\infty \ .
		$$
	\end{definition}
	
	For any $h \in F$, let $Q_h((-\infty,t]) = \mathbb{P}(x \in F : \left<x, h\right> < t) \ .$
	
	\begin{definition}[Continuous law]\label{def:continuity}
		
		We say that a probability measure $Q$ on $(F, \mathcal{F})$ is a \textit{continuous law} if for any $h \in F$ and $t \in \mathbb{R}$, $Q_h((-\infty,t])$ is continuous.
	\end{definition}
	
	We denote by $\mathcal{Q}$ the set of all continuous probability measures on $(F, \mathcal{F})$ which satisfy the Carleman condition.
	
	Given $(u,v) \in \mathcal{U}^2,$ we define a distance $D$ between $Q^u \in \mathcal{Q}$ and $Q^v \in \mathcal{Q}.$ 
	
	\begin{proposition}\label{prop:distance} 
		For all $(u,v) \in \mathcal{U}^2$, let $Q^u$ and $Q^v$ be measures over $L^2([0,T])$ satisfying the regular conditions given by Definitions \ref{def:carleman} and \ref{def:continuity}. Let also $W$ be an independent Gaussian measure on $L^2([0,T])$.
		Then
		\begin{equation}\label{eq:disteo}
			D(u,v) = \int ||f^{u,h}-f^{v,h}||_{\infty}dW(h)
		\end{equation}
		gives a metric over $\mathcal{Q}$.
	\end{proposition}
	
	\begin{proof}
		By construction $D(u,v)$ is symmetric and inherits the triangle inequality from the infinite norm, $|| \cdot ||_{\infty}.$
		Moreover, as a consequence of Theorem 4.1 in Cuesta-Albertos, Fraiman and Ransford \cite{Cuesta2007} $D(u,v)=0$  if and only if $Q^u=Q^v$.
		
	\end{proof}

	For each $u \in \mathcal{U}$ and $n \in \{ 1, \cdots, N \}$, let $Y^u_n \in F$ be a function generated by law $Q^u \in \mathcal{Q}$. For each $u \in \mathcal{U}$ we call
	\begin{equation}\label{def:sampleset}
		\mathcal{Y}^u _N= \{Y^u_n : n \in \{1,\cdots, N\}\}
	\end{equation}
	as the \textit{sample set $u$}.
	
	For every $n = 1, \ldots, N$ and every $u \in \mathcal{U}$ we call \textit{projection} of $Y^u_n$  onto direction $h$ the following inner product
	\begin{equation*}
		R^{u,h}_{n} = \int_{0}^{T}h(t)Y^u_n(t)dt \ .
	\end{equation*}
	
	For each $u \in \mathcal{U}$, the projection of the data set $u$ in the direction $h$ is naturally defined as
	\begin{equation*}
		\mathcal{Y}_N^{u,h} = \{R^{u,h}_{n}: Y^u_n \in \mathcal{Y}^u_N\}.
	\end{equation*}
	
	The \textit{empirical cumulative distribution function} of $\mathcal{Y}^{u,h}_N$ is given by
	
	\begin{equation*}
		\hat{f}^{u,h}_N(t) = \frac{1}{N}\sum_{R^{u,h}_{n} \in \mathcal{Y}^{u,h}_N}1_{\{R^{u,h}_{n} \leq t\}} \, \, , t \in \mathbb{R} \ ,
	\end{equation*}
	where $1$ denotes the indicator function. 
	
	Given $u,v \in \mathcal{U}$ consider the data sets $\mathcal{Y}^u _N$ and $\mathcal{Y}^v _N$ respectively, denote $D^{u,v}_N(h)$ to the $L_\infty$ distance between the empirical distributions of  density functions $\mathcal{Y}^{u,h} _N$ and $\mathcal{Y}^{v,h} _N,$ i.e.,
	
	\begin{equation}\label{def:distsample}
		D^{u,v}_N(h) = \sup_{t \in \mathbb{R}} \left\{|\hat{f}^{u,h}_N(t)-\hat{f}^{v,h}_N(t)| \right\}.
	\end{equation}
	
	Cuesta-Albertos, Fraiman and Ransford \cite{Cuesta2007} introduce a goodness-of-fit test for functional data based on the equation \eqref{def:distsample}, which, although has discriminating power from an asymptotic theoretical perspective in practice, can be unstable for finite sample size. Hence, a natural idea is to propose a statistic that takes many directions into account.

	\section{Bounds on the error rates}\label{sec:bounderr}

	Let $B = (B_1, \cdots, B_M)$ be $M$ independent realisations of elements in $F$ generated with a Gaussian measure. Following the reasoning of Duarte et al. \cite{duarte_retrieving_2019}, we take this measure to be the Brownian bridge. Let us define the empirical distance between the sample sets $u$ and $v$ as
	\begin{equation}\label{def:empiricaldistance}
		\hat{D}_{N,M}(u,v) = \frac{1}{M}\sum_{m = 1}^M D^{u,v}_N(B_m) \ .
	\end{equation}
	
	Let us present some useful notations. For each $(u,v) \in \mathcal{U}^2$ consider the $L_\infty$ distance between the projections in direction $B,$
	\begin{equation*}
		D^{u,v}(B) := ||f^{u,B}-f^{v,B}||_{\infty} \ ,
	\end{equation*}
	
	Denote also
	\begin{equation*}
		\Delta_N^{u,v}(B) := \Delta_N^{u}(B)+ \Delta_N^{v}(B),
	\end{equation*}
	where $\Delta_N^{u}(B)=||\hat{f}_N^{u,B}-f^{u,B}||_{\infty}$ and $\Delta_N^{v}(B)=||\hat{f}_N^{v,B}-f^{v,B}||_{\infty}.$

	While we do not have direct access to the distance $D(u,v)$, the distance $\hat{D}_{N,M}(u,v)$  gives us a finite approximation. Theorem \ref{th:bounds} shows that the probability of a large error between the estimate and the true distance decays in an exponential manner.
	
	\begin{theorem}\label{th:bounds}
		For all $(u,v) \in \mathcal{U}^2$, let $\mathcal{Y}^u_N$ and $\mathcal{Y}^u_N$ be sample sets which satisfy the regular conditions \ref{def:carleman} and \ref{def:continuity}.
		Let $\hat{D}_{N,M}(u,v)$ be defined as in equation \eqref{def:empiricaldistance}. Then, for any $\gamma \in [0,1]$, exists $C>0$ such that
		\begin{equation}\label{eq:mainbound}
			\mathbb{P}\left(|\hat{D}_{N,M}(u,v) - D(u,v)| \geq \gamma\right) \leq 2e^{\frac{-M\gamma^2}{2}}+2e^{\frac{-M\gamma^2}{32}}+2Ce^{\frac{-N\gamma^2}{16}} \ .
		\end{equation}
	\end{theorem}
	
	\begin{proof}
		
		Observe that
		
		\begin{equation*}
			\mathbb{P}(|\hat{D}_{N,M}(u,v) - D(u,v)| \geq \gamma) \leq \mathbb{P}(\hat{D}_{N,M}(u,v) - D(u,v) \geq \gamma) + \mathbb{P}(D(u,v) - \hat{D}_{N,M}(u,v)  \geq \gamma) \ .
		\end{equation*}
		
		Let us start by finding the bound of the first term, i.e. 
		
		\begin{equation*}
			\mathbb{P}\left(\frac{1}{M}\sum_{m = 1}^MD^{u,v}_N(B_m)-D(u,v)\geq \gamma\right) \ .
		\end{equation*}

		By triangle inequality, we obtain that
		\begin{equation*}
			D^{u,v}(B) \leq D^{u,v}_N(B) + \Delta^{u,v}_N(B) \ .
		\end{equation*}
		Then,
		
		\begin{eqnarray}
			& & \mathbb{P}\left(\frac{1}{M}\sum_{m = 1}^M D^{u,v}_N(B_m)-D(u,v)\geq \gamma\right) \nonumber \\ 
			&\leq& \mathbb{P}\left(\frac{1}{M}\sum_{m = 1}^M D^{u,v}(B_m) - D(u,v) + \frac{1}{M}\sum_{m = 1}^M\Delta^{u,v}_N(B_m) \geq \gamma\right) \nonumber \\ 
			&\leq& \mathbb{P}\left(\frac{1}{M}\sum_{m = 1}^M D^{u,v}(B_m) - D(u,v) \geq \gamma /2\right) + \mathbb{P}\left(\frac{1}{M}\sum_{m = 1}^M\Delta^{u,v}_N(B_m) \geq \gamma/2\right). \label{eq:separatig}
		\end{eqnarray}
		
		
		
		
		We have 
		$$
		D(u,v) = \int D^{u,v}(h)dW'(h) = \mathbb{E}_B[D^{u,v}(B)] \ ,
		$$
		where $W'$ is the Brownian bridge measure.

		Therefore, by the law of total probability, the first term in \eqref{eq:separatig} is bounded by Hoeffding's inequality \cite{hoeffding1994probability}
		\begin{equation}\label{eq:hoeffbound}
			\mathbb{P}\left(\frac{1}{M}\sum_{m = 1}^M D^{u,v}(B_m) - D(u,v) \geq \gamma/2\right) \leq e^{\frac{-M\gamma^2}{2}}
		\end{equation}
		
		The second term of \eqref{eq:separatig} can be bounded as follows,
		\begin{eqnarray}
			& & \mathbb{P}\left(\frac{1}{M}\sum_{m = 1}^M\Delta^{u,v}_N(B_m) \geq \gamma/2\right) \nonumber \\   &=& \mathbb{P}\left(\frac{1}{M}\sum_{m = 1}^M\Delta^{u,v}_N(B_m) - \mathbb{E}_B[\Delta^{u,v}_N(B)] + \mathbb{E}_B[\Delta^{u,v}_N(B)] \geq \gamma/2\right) \label{eq:lawtotalunion1}\\ 
			&\leq & \mathbb{P}\left(\frac{1}{M}\sum_{m = 1}^M\Delta^{u,v}_N(B_m) - \mathbb{E}_B[\Delta^{u,v}_N(B)] \geq \gamma/4\right) + \mathbb{P}\left(\mathbb{E}_B[\Delta^{u,v}_N(B)] \geq \gamma/4\right) \ . \nonumber
		\end{eqnarray}
		
		Inequality (\ref{eq:lawtotalunion1}) holds by the law of total union.
		
		
		
		As in equation \eqref{eq:hoeffbound} the first term is again bounded by Hoeffding
		\begin{eqnarray*}
			\mathbb{P}\left(\frac{1}{M}\sum_{m = 1}^M\Delta^{u,v}_N(B_m) - \mathbb{E}_B[\Delta^{u,v}_N(B)] \geq \gamma/4\right) \leq e^{\frac{-M\gamma^2}{32}} \ .
		\end{eqnarray*}

		Theorem 3.1 of Cuesta-Albertos, Fraiman and Ransford \cite{Cuesta2007} gives us that the distribution of $\Delta^{u,v}_N(h)$ is independent of $h$ for any $h \in F/\{0\}$. Therefore, it exits $C>0$ such that we have

		\begin{eqnarray*}
			\mathbb{P}(\mathbb{E}_B[\Delta^{u,v}_N(B)] \geq \gamma/4) = \mathbb{P}(\Delta^{u,v}_N(B) \geq \gamma/4) \leq Ce^{\frac{-N\gamma^2}{16}} \ ,
		\end{eqnarray*}
		with probability $1$ by the two sample DKW inequality given in Theorem 1 in Wei and Dudley \cite{wei2012two}.
		
		To finish, note that by the triangle inequality
		\begin{equation*}
			D^{u,v}_N(B) \leq D^{u,v}(B) + \Delta^{u,v}_N(B) \ ,
		\end{equation*}
		we also have
		
		\begin{eqnarray}
			& & \mathbb{P}\left(D(u,v)-\frac{1}{M}\sum_{m=1}^MD^{u,v}_N(B_m)\geq \gamma\right) \nonumber \\
			&\leq & \mathbb{P}\left(D(u,v)-\frac{1}{M}\sum_{m=1}^M D^{u,v}(B_m)+\frac{1}{M}\sum_{m = 1}^M\Delta^{u,v}_N(B_m)\geq \gamma\right)  \nonumber \\ \label{eq:separating2}
			&\leq & \mathbb{P}\left(D(u,v)-\frac{1}{M}\sum_{m=1}^M D^{u,v}(B_m) \geq \gamma /2\right)+\mathbb{P}\left(\frac{1}{M}\sum_{m = 1}^M\Delta^{u,v}_N(B_m)\geq \gamma/2\right).    
		\end{eqnarray}

		

		As in the previous case, the first term in equation \eqref{eq:separating2} is bounded by Hoeffding and the second is bounded following the exact same reasoning as the bound of the second term of 
		\eqref{eq:separatig}.
	\end{proof}
	
	\begin{remark}\label{rem:consC}
		The constant $C$ is discussed in Wei and Dudley \cite{wei2012two}. They show that the bound holds for any $C\geq e$, and also holds for some $C_N$ which approaches $2$ as $N$ increases. We refer to this article for the choice of $C_N$, as to obtain a more powerful statistic.
	\end{remark}

	\section{Consistency of the partition selection from a dendrogram model}
	
	The results of Section \ref{sec:bounderr} show that $\hat{D}_{M,N}$ is a consistent estimator of a distance in law between the sets given by equation (\ref{eq:disteo}) since it is a sum of exponentially decreasing terms. Therefore, we propose to use $\hat{D}_{M,N}$ as the plug-in distance for the procedure described in Section \ref{sec:clustproc}, which returns a $(P_{1:S_U}(\hat{D}_{M,N}),r_{\hat{D}_{M,N}})$ random dendrogram model, containing a family of nested partitions. Given a real-valued threshold a partition can be selected from the dendrogram. This can be shown graphically, as in Fig. \ref{fig:ExDenCut}.
	
	\begin{figure}
		\centering
		\includegraphics[scale = 0.25]{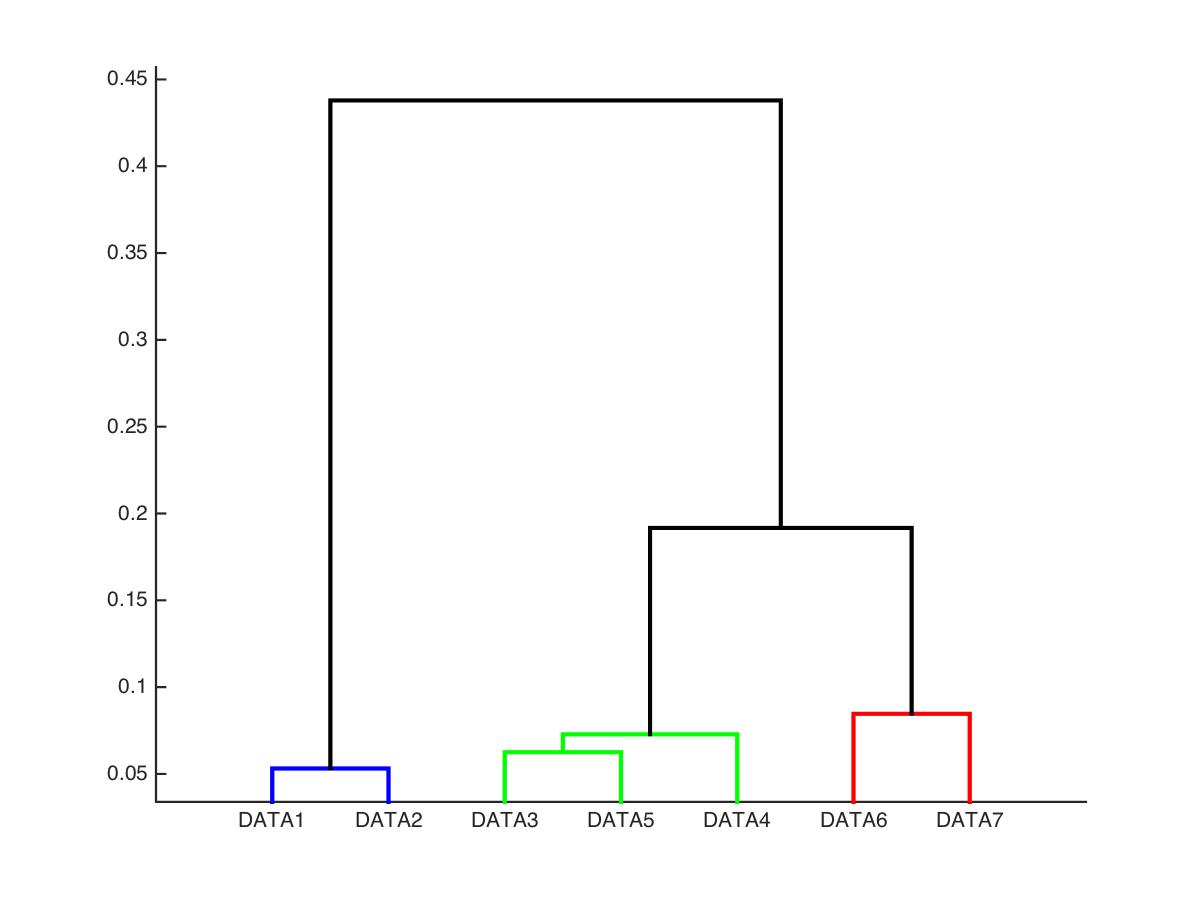}
		\caption{Graphical representation of a dendrogram model, with a specific partition selected and identified by different colors. The partition selected contains three clusters, represented by the branches in red, green, and blue.}
		\label{fig:ExDenCut}
	\end{figure}
	
	For our classification goal, we need a heuristic to choose cut threshold. We start by noting that the distance measure $\hat{D}_{M,N}$ is closely related to the statistic proposed in Cuesta-Albertos, Fraiman and Ransford \cite{cuesta2007sharp} to perform the goodness-of-fit tests.

	Cuesta-Albertos, Fraiman and Ransford  \cite{cuesta2007sharp} introduced test based on one-dimensional random projections of $\mathcal{Y}_N^u$ and $\mathcal{Y}_N^v$ to determine whether both sets were generated by the same law, i.e. testing
	
	\begin{equation}\label{hypothesis}
		H_0: Q^u=Q^v \mbox{ vs } H_A: Q^u\neq Q^v.
	\end{equation}
	
	The test statistic is
	
	\begin{equation}\label{KS-stat}
		KS\left(\hat{f}_N^{u,B},\hat{f}_N^{v,B}\right) = \sqrt{\frac{N}{2}} D_N^{u,v}(B).
	\end{equation}

	The null hypothesis is rejected at level $\alpha$ when $KS\left( \hat{f}_N^{u,B},\hat{f}_N^{v,B}\right)>\eta_{\alpha},$ the critical value is obtained from the asymptotic Kolmogorov distribution Kolmogorov \cite{an1933sulla}.

	The results presented in Section \ref{sec:bounderr} allow us to perform a test with the finite sample exponential bound on the statistic $\hat{D}_{M,N}$. In summary, we can select the partition from the dendrogram model using a criterion based on the goodness-of-fit test statistic behavior under the null hypothesis. The null hypothesis is rejected  at level $\alpha$ whenever
	\begin{equation}\label{statB}
		\hat{D}_{M,N}(u,v)>\gamma_{\alpha},
	\end{equation}
	where $\gamma_{\alpha}$ is a value which ensures the level of the test.

	Informally, this choice of threshold gives us a consistent clustering procedure for the following reasons. Let  $\mathcal{Y}_N^{u_1},\dots,\mathcal{Y}_N^{u_{S_U}}$ be random data sets generated according to $u_1,\dots,u_{S_U} \in \mathcal{U}.$ 
	Let $[u^*]$ be the equivalence class given by,
	$$
	[u^*]=\{ u  \in \mathcal{U} | Q^u=Q^{u^*}\} \ .
	$$
	Let $P^*$ be the partition given by the quotient set, where all the indexes $u$ of the data sets $\mathcal{Y}_N^{u}$ belonging to the same group of the partition have been generated by the same distribution.

	Let $\mathcal{D} \in \mathbb{R}^{S_U \times S_U}$ be the distance matrix, whose entry $(u,v)$, has the distance between $Q^{u}$ and $Q^{v}$ given by \eqref{eq:disteo}. Then $D(u,v)>0$ if and only if $Q^u\neq Q^v$, with $D(u,v)=0$ otherwise. Without loss of generality, consider a permutation of the rows and columns of $\mathcal{D}$ such that we have a block matrix, where the blocks of the diagonal are zero-squared matrices, each defining a cluster,while outside the zero-squared blocks of the diagonal the entries are positive distances. So from $\mathcal{D}$ we get the partition $P^*$.  
	
	Let $\mathcal{D}_{N} \in \mathbb{R}^{S_U \times S_U}$ be the empirical counterpart of $\mathcal{D}$, where each entry $\hat{D}_{M_N,N}(u,v)$ is the empirical distance between $\mathcal{Y}_N^{u}$ and $\mathcal{Y}_N^{v}$ as defined in \eqref{def:empiricaldistance}. This matrix is an empirical approximation of $\mathcal{D}$ and will therefore give us the correct partition for a sufficiently large $N$.

	To show a consistency result of the clustering procedure, let us give some definitions. Let $\alpha = \alpha_N$ such that
	$\alpha_N \rightarrow 0$ and  $\log(2/\alpha_N)/N\rightarrow 0$ as $N \rightarrow +\infty.$ Let $\hat{P}_{k^*}$ be the partition obtained from $\mathcal{D}_N$ using the threshold $\gamma_{\alpha_N}$, 
	$$
	\hat{P}_{k^*} = \left\{P_{K}(\hat{D}_{M,N}): K = \max_k
	\{r_{\hat{D}}(P_k(\hat{D}_{M,N})) \leq \gamma_{\alpha_N}\} \right\} \ .
	$$
	
	We also take $M$ as a function of $N$ and denote it as $M_N$. Let $M_N$ be such that $\log(2/\alpha_N)/M_N \rightarrow 0,$ as $N \rightarrow \infty$.

	\begin{theorem}
		Under the same setting stated in Theorem \ref{th:bounds}, for every $u \in \mathcal{U}$, let $\mathcal{Y}^u$ be a set of functional data assuming values in $F$ with associated law $Q^u \in \mathcal{Q},$ such that satisfy the regular conditions given by Definitions \ref{def:carleman} and \ref{def:continuity}. Then,
		$$
		\lim_{N \rightarrow \infty}\mathbb{P}(\hat{P}_{k^*} \neq P^*) = 0 \ .
		$$

	\end{theorem}
	
	\begin{proof}
		Let $U = \{(u, v) \in \mathcal{U}^2 : Q^u \neq Q^{v}\}$ and $U' = \{(u, v) \in \mathcal{U}^2 : Q^u = Q^{v}\}$. Let also $d(u,v) = \inf_{(u,v)\in U^2}\{D(u,v):u\neq v\}$. Then, for any $\beta \in (0, d(u,v))$ we define
		
		$$
		\mathcal{U}_N = \bigcup_{(u,v) \in U} \left\{\hat{D}_{M_N,N}(u, v) \leq \beta\right\},
		$$
		and
		$$
		\mathcal{U}'_N = \bigcup_{(u,v) \in U'} \left\{\hat{D}_{M_N,N}(u, v) > \beta \right\}
		$$
		
		Then
		$$
		\mathbb{P}(P_{k^*} \neq P^*) \leq \mathbb{P}(\mathcal{U}_N \cup \mathcal{U}'_N) = \mathbb{P}(\mathcal{U}_N) + \mathbb{P}(\mathcal{U}'_N) \ .
		$$
		
		Hence, it is enough to show that exists $N_0(\zeta,\beta)$ such that $\mathbb{P}(\mathcal{U}_N) < \frac{\zeta}{2}$ and $N_1(\zeta,\beta)$ such that $\mathbb{P}(\mathcal{U}'_N) < \frac{\zeta}{2}$ for every $N > N_0(\zeta,\beta)$ and $N>N_1(\zeta,\beta)$.
		
		Theorem \ref{th:bounds} gives an exponential bound between the distance  \eqref{eq:disteo} and \eqref{def:empiricaldistance}, then 
		\begin{equation*}
			\mathbb{P}\left( \hat{D}_{M_N,N}(u,v)>\gamma_{\alpha_N}\right)\leq \alpha_N,
		\end{equation*}
		if $Q^u = Q^v.$

		Moreover, we get that $\gamma_{\alpha_N} \rightarrow 0$ whenever

		$$
		\max\left\{12\sqrt{\frac{2\log\left(2/\alpha_N\right)}{M_N}} ,12\sqrt{\frac{\log\left(2C/\alpha_N\right)}{N}}\right\} \longrightarrow 0,
		$$
		as $\alpha_{N} \rightarrow 0.$  Note this convergence is guaranteed by the definitions of $\alpha_N$ and $M_N$.

		Then, whenever $Q^u \neq Q^v$, for $N$ large enough  there exists $\tilde{d}>0$ such that
		
		\begin{equation*}
			\mathbb{P}\left( \hat{D}_{M_N,N}(u,v)<\tilde{d}+\gamma_{\alpha_N}\right)\geq 1-\alpha_N ,
		\end{equation*}

		Then, for $\beta<\tilde{d}+\gamma_{\alpha_N}$
		\begin{eqnarray*}
			\mathbb{P}(\mathcal{U}_N) &\leq &  \sum_{(u,u') \in U} \mathbb{P}\Big( \hat{D}_{M_N,M}(u,v) \leq \beta \Big) \leq \\
			&\leq &  \frac{S_U (S_U-1)}{2} \alpha_N.
		\end{eqnarray*} 
		Therefore there exists $N_0(\zeta, \beta)$ such that $\mathbb{P}(\mathcal{U}_N) < \frac{\zeta}{2}$ for any $\zeta$.
		
		Also, for $N$ large enough we have $\beta > \gamma_{\alpha_N}$, and therefore 
		\begin{eqnarray*}
			\mathbb{P}(\mathcal{U'}_N) &\leq &  \sum_{(u,u') \in U'} \mathbb{P}\Big( \hat{D}(u,v) > \beta \Big) \leq \\
			&\leq &  \frac{S_U (S_U-1)}{2} \alpha_N.
		\end{eqnarray*} 
		Then there exists $N_1(\zeta, \beta)$ such that $\mathbb{P}(\mathcal{U'}_N) < \frac{\zeta}{2}$ for any $\epsilon$.
		
		Then, for $N$ large enough the clustering procedure with $k$ clusters will retrieve a partition which is coincident with $P^*$ with arbitrarily high probability.

	\end{proof}
	
	\section{An empirical Bernstein bound for threshold selection}\label{sec:bernthres}

	The proof of the theorem \ref{th:bounds} uses Hoeffding's inequality and relies on the fact that the random variables  $D_N^{u,v}(B)$ are bounded in $[0,1].$ However, if the variance of $D_N^{u,v}(B)$ is small compared to the range of the bound, a sharper bound can be obtained using a Berstein-type inequality. This strategy presents a challenge, since the variance of $D_N^{u,v}(B)$ is unknown. To overcome this problem, we propose to use an empirical Berstein inequality, introduced in Maurer and Pontil \cite{maurer2009empirical}, which uses the empirical variance of $D_N^{u,v}(B)$ instead of the theoretical one.
	
	Using this strategy, let us show that under $H_0$ we have an alternative bound for the estimation error. First, let us define the following ancillary variables. Let $\hat{V}_B[\bar{D}^{u,v}_N]$ be the empirical variance in $B$ of $D^{u, v}_N(B)$, that is
	$$
	\hat{V}_B[\bar{D}^{u,v}_N] =  \frac{1}{M-1}\sum_{m=1}^M \left[D^{u, v}_N(B_m)-\bar{D}^{u,v}_N \right]^2  \ ,
	$$
	where $\bar{D}^{u,v}_N = \frac{1}{M}\sum_{m=1}^MD^{u, v}_N(B_m)$. Let also, for any $\delta \in (0,1)$,
	$$
	\Gamma^{u,v}(\delta) = \sqrt{\frac{2\hat{V}_B[\bar{D}^{u,v}_N]log(2/\delta)}{M}},
	$$
	and
	$$
	\epsilon(\delta) = \frac{7\log(2/\delta)}{3(M-1)} \ .
	$$
	\begin{theorem}\label{th:boundH0Bernstein}
		Let $\mathcal{Y}^u_N$ and $\mathcal{Y}^v_N$ be two sample sets which follow the regular conditions \ref{def:carleman} and \ref{def:continuity}. Let also $Q^u=Q^v$.   Then for any $\gamma \in (0,1)$ and $\delta$ such that $\epsilon(\delta)<\gamma$ we have
		$$
		\mathbb{P}(|\hat{D}_{M,N}(u,v) - D(u,v)| \geq \Gamma^{u,v}(\delta)+\gamma) \leq C\exp(-N(\gamma-\epsilon(\delta))^2)+\delta
		$$
		where $C$ is the same constant discussed in Remark \ref{rem:consC}.

	\end{theorem}
	\begin{proof}
		Under $H_0$ we have $D(u,v) = 0$. Therefore the probability simplifies to
		$$\mathbb{P}\left(\frac{1}{M}\sum_{m=1}^M D^{u,v}_{N}(B_m) \geq \Gamma^{u,v}(\delta)+\gamma\right).
		$$
		
		Then we have,
		\begin{eqnarray}\label{eq:unionH0}
			\mathbb{P}\left(\frac{1}{M}\sum_{m=1}^MD^{u,v}_{N}(B_m) +\mathbb
			E_B[D^{u,v}_{N}(B)] - \mathbb{E}_B[D^{u,v}_{N}(B)]\geq \Gamma^{u,v}(\delta)+\gamma-\epsilon(\delta)+\epsilon(\delta)\right) \leq \\ \nonumber  \leq \mathbb{P}\left(\frac{1}{M}\sum_{m =1}^{M}D^{u,v}_{N}(B_m) - \mathbb{E}_B[D^{u,v}_{N}(B)]\geq \Gamma^{u,v}(\delta)+\epsilon(\delta)\right) + \mathbb{P}\left(\mathbb
			E_B[D^{u,v}_{N}(B)]\geq \gamma-\epsilon(\delta)\right) \ .
		\end{eqnarray}
		Now note that 
		\begin{eqnarray*}
			\mathbb{P}\left(\mathbb
			E_B[D^{u,v}_{N}(B)]\geq \gamma-\epsilon(\delta)\right) \leq \mathbb{P}\left(\mathbb
			E_B[\Delta^{u,v}_{N}(B)]\geq \gamma-\epsilon(\delta)\right) \leq Ce^{-N(\gamma-\epsilon(\delta))^2}
		\end{eqnarray*}
		by the same arguments used in the proof of Theorem \ref{th:bounds}.
		The first term of \eqref{eq:unionH0} is bounded by $\delta$ following the empirical Berstein's inequality of Maurer and Pontil  \cite{maurer2009empirical} which achieves the upper bound and finishes the proof.

	\end{proof}
	Since we have access to $\Gamma^{u,v}(\delta)$ from the data for all $(\mathcal{Y}^u_N,\mathcal{Y}^v_N)$, this bound gives us a heuristic for choosing a partition from the dendrogram model. We control the number of random directions and suggest choosing them such that $M > N$. This results in the bound obtained in Theorem \ref{th:boundH0Bernstein} being dominated by its first term, giving us a bound close to a DKW-type bound Massart \cite{massart1990tight}.
	
	The random variables $\hat{V}_B[\bar{D}^{u,v}_N]$ do not necessarily have the same value for all pairs $(u,v)$. To deal with this, we select our threshold using the worst case among all the comparisons, this means taking the maximum variance among all the  pairs $(u,v)$. Let $\Bar{V}^*$ be
	$$
	\Bar{V}^* = \max_{(u,v) \in \mathcal{U}^2}\left\{\hat{V}_B[\bar{D}^{u,v}_N]\right\} \ ,
	$$
	then replace $\Gamma^{u,v}(\delta)$ by
	$$
	\Gamma^*(\delta) = \sqrt{\frac{2V^*\log(2/\delta)}{M}}.
	$$
	The following threshold controls $\hat{D}_{M,N}(u,v)$ for all pairs $(u,v)\in\mathcal{U}^2$,
	
	\begin{definition}\label{def:gammaN}
		$$\gamma_{\alpha_N}^* = \inf_{\delta\in(0,\alpha_N)}\left\{\Gamma^*(\delta)+\sqrt{\frac{\log(\frac{C}{\alpha_N-\delta})}{N}}+\epsilon(\delta)\right\} \ .$$
	\end{definition}

	This value can be approximated numerically.

	\section{Pseudo-code and simulations}\label{sec:simulations}
	
	In this section, we show a pseudo-code of the clustering algorithm and also give some considerations on the numerical aspects of the method. All codes used can be found at \url{ https://github.com/fanajman/Clustering-Sets-of-Functional-Data-by-Law}.
	
	\begin{algorithm}[H]\label{alg:cluster}
		\KwData{Functional data sets \{$\mathcal{Y}^u:u\in \mathcal{U}\}$
			\\ Independent directions generated following a Brownian bridge \ $(B_1,\cdots,B_M)$
			\\ Level $\alpha_N$
		}

		initialization\;
		\For{$u \in \mathcal{U}$}{
			\For{$m \in (1,\cdots,M)$}{
				\For{$n\in(1,\cdots,N)$}{
					$\mathcal{R}^{u,B_m}_{n} = <Y_n^u, B_m>$\;}
				$f^{u,B_m}_N(t) = \frac{1}{N}\sum_{n = 1}^{N}1_{\{R^{u,B_m}_{n} \leq t\}}$\;
			}
		}
		\For{$u\in\mathcal{U}$, $v \in \mathcal{U}$, $u \neq v$}{
			\For{$m \in (1,\cdots,M)$}{
				$D_N^{u,v}(B_m) = \sup_{t \in \mathbb{R}}\{|f^{u,B_m}_N(t)-f^{v,B_m}_N(t)|\}$\;}
			$\hat{D}_{M,N}(u,v) = \frac{1}{M}\sum_{m = 1}^M D^{u,v}_N(B_m)$\;
			$\hat{V}_B[\bar{D}^{u,v}_N] = \frac{1}{M-1}\sum_{m = 1}^M(D^{u,v}_N(B_m)-\hat{D}_{M,N}(u,v))^2$\;
		}
		$V^* = \sup_{u,v}\{\hat{V}_B[\bar{D}^{u,v}_N]\}$\;
		$\gamma^*_{\alpha_N} \text{as defined in \ref{def:gammaN}}$\;
		$P_0 = \{\{u\}:u\in \mathcal{U}\}$\;
		\For{$k \in (1,\cdots,|\mathcal{U}|)$}{
			\For{$C \in P_{k-1}$, $C' \in P_{k-1}$, $C\neq C'$}{
				$D(C,C') = \max\{\mathcal{D}_N(u,v): u\in C, v\in C'\}$\;
			}
			$(C^k_1,C^k_2) = \arg\inf_{(C_1,C_2)}\{D(C_1,C_2)<D(C,C'):C\in P_{k-1}, C'\in P_{k-1}, C\cap C' = \emptyset\}$\;
			$r_k =  D(C^k_1,C^k_2) $\;
				$C^k_{1,2} = C^k_1\cup C^k_2$\;
				$P_k = \{C \in P_{k-1} : C \neq C_k^1, C \neq C^k_1\} \bigcup C^k_{1,2}$\;
			\If{$r_k \geq \gamma^*_{\alpha_N}$}
			{$\hat{P}_{k^*} = P_{k-1}$\;
				$\textbf{Break}$}
		}
		
		$\textbf{Return} \ \  \hat{P}_{k^*}$
	\end{algorithm}

	In the following we will present a simulation study to show the performance of the clustering procedure. For this purpose, we propose two models to generate functional data sets.
	
	We call the first model a $\theta$-scaled Brownian Bridge (SBB). Each sample was generated independently as follows 
	$$
	Y^u_n(t) = \left(W(t)-\frac{t}{T}W(T)\right)\theta_u \ ,
	$$
	where $W$ is the Wiener process and $\theta_u \in \mathbb{N}$ is the fixed parameter defining of the $\theta$-scale of the Brownian Bridge in this simulation.
	
	We call the second model is Autoregressive (AR). Each sample was generated independently as
	$$
	Y^u_n(t) = Y^u_n(t-1)\theta'_u + \xi_n \ ,
	$$
	where $Y^u_n(0) = \xi_0$ and $\xi_n$ are independent random variables generated following a standard normal distribution, and $\theta'_u \in (0,1)$ is the model parameter.
	In both models, the functions were generated on an equispaced grid of 80 points.

	Seven functional data sets were generated for each model, these data sets followed three different laws, which means that there are three clusters. For SBB, the parameters $\theta$ chosen for each of the seven data sets are $(1,1,2,2,2,4,4)$. While for AR the parameters $\theta'$ chosen for each of the seven models are $(0.99,$ $0.99,$ $0.66,$ $0.66,$ $0.66,$ $0.33,$ $0.33)$. Then, in both cases, there are two clusters consisting of two sets of functional data, while the remaining one has three sets of functions.
	
	To initialise the clustering procedure we need to set two parameters, namely $\alpha_N = \sqrt{1/N}$ and $M_N = \sigma N$. For each different combination of sample size $N \in (40, 60, \cdots, 160)$, $\sigma \in (10,30,50)$ and model, we ran and retrieved a partition $\hat{P}_{k^*}$ for 100 replicates.

	Figures \ref{fig:per_brow} and \ref{fig:per_ar} show that the performance of the clustering procedure is very satisfactory. The partition retrieved from the SBB model data was correct in more than $90\%$ of the replicates for  $N\geq60$ and every value of $\sigma$. The case of the AR model is more challenging, achieving the correct partition in more than $85\%$ for $N\geq 80$. Interestingly, the best result for the AR model is obtained for $\sigma = 10$.  However, most of the results obtained for the different values of $\sigma$ are similar for both models, indicating that increasing the value of $M_N$ above $10N$ does not bring clear benefits.

	\begin{figure}[h]
		\centering
		\includegraphics[scale = 0.5]{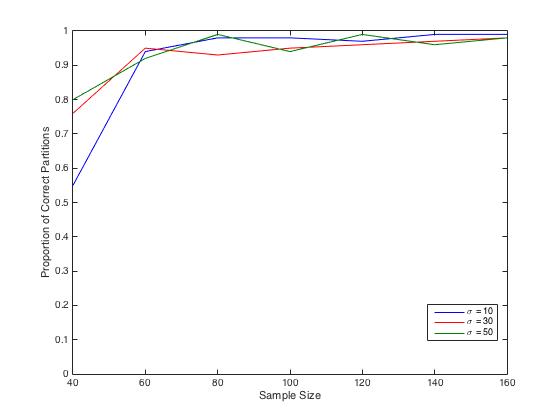}
		\caption{SBB model: The proportion of partitions correctly retrieved with data generated following the $\theta$-scaled Brownian bridge model for each sample size $N = (40,60,\cdots,160)$ and for each $\sigma \in (10, 30, 50)$.}
		\label{fig:per_brow}
	\end{figure}
	
	\begin{figure}[ht]
		\centering
		\includegraphics[scale = 0.5]{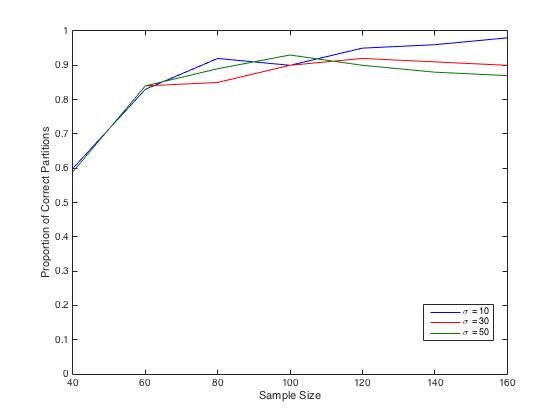}
		\caption{AR model: The proportion of partitions correctly retrieved with data generated following the Auto regressive model for each sample size $N = (40,60,\cdots,160)$ and for each for $\sigma \in (10, 30, 50)$.}
		\label{fig:per_ar}
	\end{figure}

	Finally, to better understand the procedure, we will characterize the errors made in cases where the partition obtained does not coincide with $P^*$. This procedure is based on a goodness-of-fit test, therefore there are two possible errors: on the one hand, the algorithm could assign to the same cluster a pair of data sets generated with different probability distributions. On the other hand, the algorithm could assign to different clusters two sets generated under the same probability law. We call the first a type $1$ error and the second a type $2$ error.  Fig. \ref{fig:e0} (resp. \ref{fig:e1}) shows the results for the error of type $1$ (resp. type $2$), for both models and $\sigma=10$.
	
		\begin{figure}[ht]
		\centering
		\includegraphics[scale=0.5]{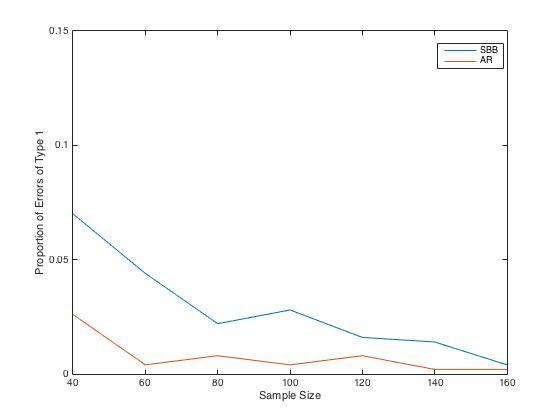}
		\caption{Proportion of data sets generated with the same law assigned to different clusters for the replicates generated with each model and $\sigma = 10$.}
		\label{fig:e1}
	\end{figure}
	
	\begin{figure}[ht]
		\centering
		\includegraphics[scale=0.5]{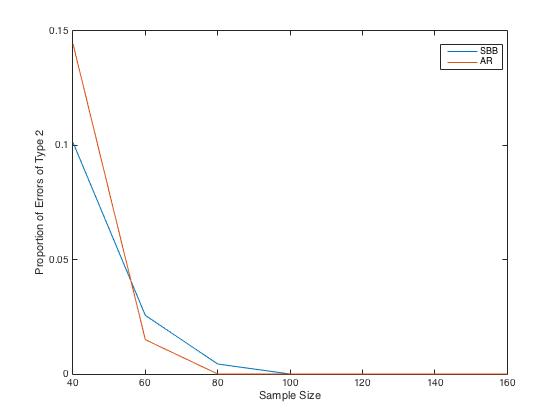}
		\caption{Proportion of data sets generated with different law assigned to the same cluster for the replicates generated with each model and $\sigma = 10$.}
		\label{fig:e0}
	\end{figure}

	As expected, the incidence of errors of type $2$ decreases with increasing sample size. For sample sizes greater than 100, there are no type $2$ errors for any model. For type 1 errors, a small decay pattern is observed, specially with the SBB model. We can also observe that while we obtain a small percentage of incorrect partitions for some combinations of parameters, the final partitions correctly assign most data pairs even for small values of $N$.

	To conclude, the clustering procedure was able to retrieve the correct partition from the set of functional data sets with realistic sample sizes. The high dimensionality of the data did not pose a significant challenge to the procedure and was in line with what was expected following the theoretical bound presented in Sections \ref{sec:bounderr} and \ref{sec:bernthres}. 
	
	The lack of assumptions on the law generating the data also makes this method useful for exploratory analysis, being able to retrieve the distribution that generated the models correctly identifying the clustering structure. The simplicity of the procedure and the heuristics of the partition selection method also make this procedure easy to use, minimizing the need for parameter tuning or fitting. This helps to avoid user-induced errors in the analysis.

	\section*{Acknowledgments}
		This work is part of the activities of FAPESP Research, Innovation and Dissemination Center for Neuromathematics (grant $\#$ 2013/ 07699-0 , S.Paulo Research Foundation (FAPESP). This work is supported by CAPES (88882.377124/2019-01) and FAPESP (2022/00784-0) grants. A.G and C.D.V. were partially supported by CNPq fellowships (grants 314836/2021-7 and 310397/2021-9) This article is also supported by FAPERJ ( $\#$ CNE 202.785/2018 and $\#$ E- 26/010.002418/2019), and FINEP ( $\#$ 18.569-8) grants. 
		
		The authors acknowledge the hospitality of the Institut Henri Poincar\'e (LabEx CARMIN ANR-10-LABX-59-01) where part of this work was written.


\clearpage

\begin{thebibliography}{99}
	
\bibitem{Cuesta2007} CUESTA-ALBERTOS, J. A., FRAIMAN, R. and RANSFORD, T. (2007a). Random projections and goodness-of-fit
tests in infinite-dimensional spaces. \textit{Bulletin of the Brazilian Mathematical Society, New Series} \textbf{ 37} 477-501.

\bibitem{cuesta2007sharp}CUESTA-ALBERTOS, J. A., FRAIMAN, R. and RANSFORD, T. (2007b). A sharp form of the Cramér–Wold theorem. \textit{Journal of Theoretical Probability} \textbf{20} 201–209	

\bibitem{duarte_retrieving_2019} DUARTE, A., FRAIMAN, R., GALVES, A., OST, G. and VARGAS, C. D. (2019). Retrieving a Context Tree from
EEG Data. \textit{Mathematics} \textbf{7} 427.

\bibitem{hoeffding1994probability}HOEFFDING, W. (1994). Probability inequalities for sums of bounded random variables. \textit{The collected works of Wassily Hoeffding} 409–426.


\bibitem{an1933sulla}KOLMOGOROV, A. (1933). Sulla determinazione empirica di una legge didistribuzione. \textit{Giorn Dell’inst Ital Degli
Att} \textbf{4} 89–91.

\bibitem{massart1990tight}MASSART, P. (1990). The tight constant in the Dvoretzky-Kiefer-Wolfowitz inequality. \textit{The annals of Probability}
1269–1283.

\bibitem{maurer2009empirical} MAURER, A. and PONTIL, M. (2009). Empirical bernstein bounds and sample variance penalization. \textit{arXiv
preprint arXiv:0907.3740}

\bibitem{mora2015adaptive} MORA-L\'OPEZ, L. and MORA, J. (2015). An adaptive algorithm for clustering cumulative probability distribution
functions using the Kolmogorov–Smirnov two-sample test. \textit{Expert Systems with Applications} \textbf{42} 4016–4021.

\bibitem{naaman2021tight} NAAMAN, M. (2021). On the tight constant in the multivariate dvoretzky–kiefer–wolfowitz inequality. \textit{Statistics
	\& Probability Letters} \textbf{173} 109088.


\bibitem{wei2012two}WEI, F. and DUDLEY, R. M. (2012). Two-sample dvoretzky–kiefer–wolfowitz inequalities. \textit{Statistics \& Probabil-
ity Letters} \textbf{82} 636–644.

\bibitem{zhu2021clustering}ZHU, Y., DENG, Q., HUANG, D., JING, B., ZHANG, B. et al. (2021). Clustering based on Kolmogorov–Smirnov
statistic with application to bank card transaction data. \textit{Journal of the Royal Statistical Society Series C} \textbf{70} 558–578.




	
\end{thebibliography}
\end{document}